\documentclass[journal]{IEEEtran}
\usepackage{amsmath,amssymb,amsfonts}
\usepackage{graphicx}
\usepackage{booktabs}
\usepackage{multirow}
\usepackage{array}
\usepackage{makecell}
\usepackage{tabularx}
\usepackage{algorithm}
\usepackage{algpseudocode}
\usepackage{url}
\usepackage{xcolor}
\usepackage{cite}
\usepackage{stfloats}
\usepackage{amsthm}
\usepackage{adjustbox}
\usepackage{microtype}
\usepackage{enumitem}

\newcommand{\algname}{SIP-MOSP}
\newcommand{\algbm}{SIP-MOSP-BM}
\newcommand{\algst}{SIP-MOSP-ST}
\newcommand{\cost}{\mathbf{c}}
\newcommand{\gvec}{\mathbf{g}}
\newcommand{\fvec}{\mathbf{f}}
\newcommand{\hvec}{\mathbf{h}}
\newcommand{\R}{\mathbb{R}}
\newcommand{\front}{\mathcal{F}}
\newcommand{\snap}{\mathcal{S}}
\newcommand{\base}{\mathcal{B}}
\newcommand{\deltaf}{\Delta}
\newcommand{\idx}{\mathcal{I}}
\newcommand{\nil}{\bot}

\newcommand{\dom}{\preceq}

\newtheorem{theorem}{Theorem}
\newtheorem{lemma}{Lemma}

\begin{document}

\title{Asynchronous Parallel Search for Exact Multi-Objective Shortest Paths with Versioned Frontier Snapshots and Indexed Dominance Pruning}

\author{Xiaoqing Xu,~Ning Zhang,~Liuyihui Qian,~Xiaojun Liu,~Juan Wu,~and~Hong Tang%
\thanks{The authors are with China Telecom Research Institute, Guangzhou, China.}%
\thanks{Corresponding authors: Hong Tang and Juan Wu (e-mail: tangh@chinatelecom.cn, wuj55@chinatelecom.cn).}%
}

\markboth{Preprint}%
{Xu \MakeLowercase{\textit{et al.}}: Asynchronous Parallel Search for Exact Multi-Objective Shortest Paths}

\maketitle

\begin{abstract}
Exact multi-objective shortest-path (MOSP) search computes the complete Pareto set between specified start and goal vertices, and its computational cost can grow rapidly with expanding nondominated label sets and frequent dominance tests over per-vertex Pareto frontiers. Efficiently parallelizing exact MOSP remains an open challenge. This paper presents \algname{} (Snapshot-based Indexed-Pruning MOSP), an asynchronous exact framework that separates label expansion from frontier maintenance within a single cooperative search. SIP-MOSP combines immutable versioned frontier snapshots with indexed dominance pruning, enabling concurrent label processing without concurrent access to the same mutable frontier. Together, these mechanisms reduce synchronization overhead and accelerate dominance testing. We instantiate the framework with block-minimum (\algbm{}) and segment-tree-minimum (\algst{}) indices and prove exactness. We evaluate both variants against four state-of-the-art exact MOSP baselines covering sequential and parallel search. Experiments across multiple objective dimensions on a road network, an Internet service provider topology, and an 180-vertex complete directed graph show that SIP-MOSP achieves speedups of up to 46.9$\times$ over the best-performing sequential baseline and up to 7.05$\times$ over the best-performing parallel baseline on mutually solved instances. In the 20-objective complete-graph setting, where many instances remain unsolved by the sequential baselines within one hour, SIP-MOSP-ST achieves a 3.34$\times$ speedup while reducing peak memory by a factor of 60.3 relative to the best-performing parallel baseline. These results demonstrate that SIP-MOSP is an efficient shared-memory framework for exact MOSP across structurally diverse graph topologies.
\end{abstract}


\section{Introduction}
\IEEEPARstart{M}{ulti-objective} shortest path (MOSP) search arises when a route must be evaluated by several criteria rather than by a single scalar metric. In transportation networks, these criteria may include travel time, distance, exposure, and the number of intersections~\cite{Bronfman2015,Funke2017}. In communication networks, routing objectives may include delay, bandwidth-related capacity, reliability, congestion exposure, energy, and economic cost. To accommodate such diverse requirements, multi-criteria routing protocols and architectures have been developed~\cite{Sobrinho2020,Tabaeiaghdaei2025}. Prior work on telecommunication network design has likewise emphasized that heterogeneous QoS and technical--economic metrics often motivate explicit multicriteria routing models rather than a priori scalarization~\cite{Craveirinha2024,Chen2015}. A scalar shortest path is computationally convenient, but it hides trade-offs that can be essential for path computation and service placement.

Exact MOSP search is difficult because it must preserve every non-dominated trade-off~\cite{Breugem2017,Serafini1987}. The number of Pareto-optimal start--goal cost vectors may grow exponentially with the number of objectives and the graph size. Moreover, the difficulty is not limited to the final output size. During the search, an exact solver maintains a nondominated label frontier at each reached vertex. Labels at non-goal vertices represent alternative start-to-vertex paths, whereas the frontier at the goal contains the currently known nondominated solutions. As these sets grow, dominance queries, insertion of surviving labels, and deletion of dominated labels become major computational costs. Classical algorithms, including Martins' method~\cite{Martins1984}, expose this frontier-growth difficulty. Modern one-to-one solvers improve practical performance through admissible heuristics, dimensionality reduction, lazy checks, and specialized frontier structures, including NAMOA*~\cite{Mandow2010}, NAMOA*dr~\cite{Pulido2015}, EMOA*~\cite{Ren2022}, LTMOA*/LazyLTMOA*~\cite{Hernandez2023}, NWMOA*~\cite{Ahmadi2024}, and T-MDA~\cite{Maristany2023}.

Parallelizing this expensive search introduces two additional challenges: many exact best-first MOSP algorithms rely on globally ordered work selection, while concurrent workers must coordinate access to evolving Pareto frontiers. Classical label-setting algorithms repeatedly extract a globally minimal or globally Pareto-optimal label from an ordered OPEN structure. Removing several labels solely according to a scalar total order can destroy the label-setting property, whereas a fully multidimensional Pareto queue is not known to have an efficient general implementation for three or more objectives~\cite{Sanders2013}. Existing parallel designs therefore expose concurrency in different ways. Sanders and Mandow~\cite{Sanders2013} process globally Pareto-optimal batches and parallelize node-local operations, with their most efficient Pareto-queue realization specialized to two objectives. Parallel MOA* runs complete searches under different objective orders and shares upper bounds, but each search retains its own open lists and frontier state~\cite{Ahmadi2025}. SOPMOA* lets multiple sub-searchers cooperate through one shared OPEN queue and shared per-vertex Pareto fronts, while locks and shared/exclusive locking coordinate frequent queue and frontier accesses~\cite{Truong2025}. OPMOS~\cite{Gold2025} processes multiple order-eligible labels concurrently, and MPMOS~\cite{Gold2026MPMOS} extends this ordered-search model to GPU-scale batches. These methods respectively rely on Pareto-optimal batches, replicated order-specific searches, synchronized shared mutable state, or specialized ordered execution. In a different problem setting, DynaMOSP performs parallel heuristic path updates after changes in large dynamic networks, but returns one user-preference-guided Pareto-optimal or suboptimal path rather than the complete Pareto frontier~\cite{Shovan2025}. The complementary question addressed here is whether one cooperative exact search can exploit concurrent label expansion and parallel maintenance of different vertex frontiers while preventing concurrent writes to the same frontier.

Our earlier DAMPC work~\cite{Xu2025DAMPC} separated candidate expansion from frontier admission through an asynchronous two-stage pipeline. DAMPC showed that assigning graph vertices to fixed update threads avoids conflicting writes to the same vertex frontier, but it relied on a centralized dispatcher and did not publish frontiers as versioned, indexed, read-optimized snapshots. Its published state also did not distinguish a stable indexed portion from recent additions. The present work retains the useful expansion/update separation but redesigns communication and frontier management for exact goal-directed MOSP.

We present \algname{}, an asynchronous shared-memory framework for exact additive MOSP. Search workers and update workers jointly advance one search while performing complementary tasks. Search workers expand accepted labels, generate successor candidates, and retain local best-first order in worker-local priority queues. Update workers validate candidates and maintain the frontiers of assigned intermediate vertices; a dedicated goal-update worker maintains the goal frontier. Each worker receives concurrent arrivals through a worker-specific incoming queue and alone transfers them into its local priority queue. Consequently, multiple producers can deliver labels concurrently, while only the receiving worker modifies the priority queue that determines its local processing order. A search worker routes each candidate directly to the update worker assigned to the candidate label's endpoint, removing the centralized dispatcher used by DAMPC.

To separate read-intensive pruning from mutable updates, every maintained frontier has two representations. The assigned update worker modifies a \emph{canonical frontier}, defined as the exact current nondominated set. Search workers never read this mutable set directly. Instead, the update worker publishes immutable, versioned snapshots. A large snapshot consists of a consolidated lexicographically ordered base, a small recent delta containing additions since the last consolidation, and an index built over the base. The delta makes a newly accepted vector visible immediately, while base reconstruction is deferred. A candidate label that survives a snapshot query carries the observed version to the assigned update worker. If the version is unchanged, the previous negative result, meaning that no snapshot vector weakly dominates the candidate label, remains valid and the duplicate query is skipped. Otherwise, the current snapshot is queried again.

We develop two exact snapshot indexes. \algbm{} stores component-wise minima for fixed-size blocks of the base. \algst{} stores the same type of minima in a segment tree over base intervals. Both first restrict the lexicographically ordered base to the prefix that can contain a dominator and then use minima only to discard impossible ranges; all surviving candidates are tested by exact vector dominance. Thus, the indexes change the amount of work, not the semantics of pruning.

The contributions are threefold.
\begin{itemize}[leftmargin=*]
    \item We introduce an asynchronous exact MOSP architecture in which search workers and update workers cooperate within one search. Worker-specific incoming queues separate concurrent label delivery from local best-first processing, direct endpoint-owner delivery removes centralized redistribution, and assigning every vertex to one update worker prevents concurrent updates to the same frontier while allowing different vertex frontiers to be updated in parallel.
    \item We develop versioned base--delta frontier snapshots and two exact dominance indexes. The design publishes recent frontier information without rebuilding a large index after every insertion, supports block-minimum and segment-tree-minimum pruning with exact verification, and safely reuses a negative snapshot result when its version remains current.
    \item We establish exactness under asynchronous expansion and evaluate the framework against four exact baselines on three topology families. The experiments quantify solve time and memory, isolate the contribution of each mechanism, and examine thread scaling and processor utilization.
\end{itemize}

The remainder of the paper is organized as follows. Section~\ref{sec:related} reviews exact and parallel MOSP search and dominance checking. Section~\ref{sec:problem} defines the problem. Section~\ref{sec:framework} presents \algname{}. Section~\ref{sec:correctness} gives correctness and complexity arguments. Section~\ref{sec:evaluation} reports the experiments, and Section~\ref{sec:conclusion} concludes.

\section{Related Work}
\label{sec:related}

\subsection{Exact Multi-Objective Shortest Path Search}

Exact MOSP algorithms compute the complete set of nondominated start--goal cost vectors and representative paths rather than a single scalarized route~\cite{Garroppo2010,Salzman2023Survey,Salzman2026Emerging}. Classical label-setting and label-correcting algorithms, represented by Martins' algorithm~\cite{Martins1984} and related methods~\cite{Skriver2000,Guerriero2001}, establish the basic rule used by later methods: a label can be discarded only if another feasible label reaching the same vertex is no worse in every objective. Goal-directed algorithms add lower bounds to prioritize labels and to prune labels whose best possible completion is already dominated.

NAMOA* formalized multi-objective A* search with consistent vector heuristics~\cite{Mandow2010}. NAMOA*dr showed that lexicographic ordering allows dimensionality reduction in local dominance checks~\cite{Pulido2015}. EMOA* organizes projected frontiers in balanced binary search trees~\cite{Ren2022}. LTMOA* and LazyLTMOA* show that contiguous arrays and delayed checks can outperform more complex structures in practice~\cite{Hernandez2023}. NWMOA* stores truncated frontiers in lexicographic order, uses early termination of frontier scans, and supports bounded negative-weight instances~\cite{Ahmadi2024}. T-MDA provides another efficient targeted multiobjective label-setting design~\cite{Maristany2023}. These methods primarily optimize the amount and cost of sequential frontier work. \algname{} instead asks how the same exact work can be organized for high-throughput shared-memory execution without making each mutable frontier a concurrent data structure.

\subsection{Parallel MOSP}

Sanders and Mandow~\cite{Sanders2013} introduced a parallel label-setting method that scans subsets of globally Pareto-optimal labels and parallelizes candidate generation, grouping, and node-local merging. Their analysis is particularly strong for two objectives, where Pareto queues and ordered local sets admit efficient geometric operations. For $d\geq3$, the paper notes that an efficient Pareto queue is not known, limiting the appeal of the fully parallel formulation.

Parallel MOA*~\cite{Ahmadi2025} executes one complete MOA*-style search for each selected objective order and shares upper-bound information. In the published implementation, $d$ cyclic orders are mapped to $d$ CPU threads. This portfolio-style design is effective when different orders expose complementary solutions early, but each search maintains a separate open list and node-frontier state. Its natural concurrency therefore follows the number of selected orders.

SOPMOA*~\cite{Truong2025} follows a shared-state design. Multiple sub-searchers remove labels from one shared OPEN queue and query or update shared per-vertex Pareto fronts. Locks protect OPEN and the solution set, while shared/exclusive locking coordinates frontier queries and updates. The design avoids replicating complete searches, but synchronization remains on queue and frontier operations that occur repeatedly throughout the search. In contrast, \algname{} uses worker-local priority queues for processing order, immutable snapshots for read-side pruning, and assigns each vertex frontier to one update worker.

OPMOS~\cite{Gold2025} couples an ordered search algorithm with specialized architecture so that several eligible path extensions can be processed concurrently while retaining the ordering discipline needed for work-efficient search. MPMOS~\cite{Gold2026MPMOS} builds on this direction with a GPU-based massively parallel architecture.

DynaMOSP~\cite{Shovan2025} addresses incremental path recomputation in large fully dynamic networks. Its parallel heuristic update returns one route selected according to user preferences, which may be Pareto-optimal or suboptimal, rather than enumerating the complete cost-unique Pareto frontier. It therefore solves a different problem from the exact complete-frontier search considered here.

DAMPC~\cite{Xu2025DAMPC} decouples expansion and path-set updates and assigns node partitions to update threads. That work established the usefulness of read/update specialization, but it relied on centralized dispatch and did not provide versioned indexed frontier snapshots. In particular, it lacked (i) a read-only frontier representation suitable for lock-free pruning, (ii) a base--delta policy that makes new vectors visible without rebuilding a large index, and (iii) validation reuse based on snapshot versions. \algname{} also replaces centralized dispatch with direct endpoint-owner delivery. The combination of a multi-producer/single-consumer incoming queue and a private best-first queue is central to this change: many threads can append labels to an update worker's incoming queue, while only that worker inserts them into its ordered queue and modifies its assigned frontiers.

\subsection{Dominance Checking and Frontier Pruning}

Dominance checking is among the most frequent operations in exact MOSP. A generated label can be screened against the nondominated costs already stored for its endpoint before admission. Separately, its lower-bound vector can be screened against the goal's Pareto frontier before expansion; if a complete solution weakly dominates that lower bound, no continuation can yield a new Pareto-optimal goal cost. Arrays, lists, balanced trees, sorted frontiers, quick checks, and SIMD layouts have all been proposed to reduce this cost~\cite{Ren2022,Hernandez2023,Hernandez2024SIMD,Ahmadi2024}.

Concurrency changes the cost model. If many threads directly query and modify the same mutable frontier, exact reads must be coordinated with insertions and deletions, either by locks that can delay readers during updates or by a substantially more complex concurrent data structure. \algname{} avoids these alternatives. Updates are exclusive at the vertex-assignment level, whereas queries use immutable snapshots. The published snapshot may retain vectors that have since become dominated, but such vectors remain safe for positive pruning because dominance is transitive. An older snapshot may miss a newly available pruning opportunity, but it cannot create an invalid solution. Version validation addresses the only case in which a previously negative result is reused for admission.

\section{Problem Definition and Preliminaries}
\label{sec:problem}

\subsection{Exact Additive MOSP}

Let $G=(V,E)$ be a finite directed graph. Each edge $e\in E$ has a nonnegative $d$-dimensional cost vector
\[
    \cost(e)=(c_1(e),\ldots,c_d(e))\in\R_{\geq0}^{d}.
\]
For a path $\pi$, its cost is the component-wise sum
\[
    \cost(\pi)=\sum_{e\in\pi}\cost(e).
\]
For vectors $a,b\in\R^d$, $a\dom b$ denotes weak dominance: $a_i\le b_i$ for all $i$. Strict Pareto dominance additionally requires $a\ne b$. Given start vertex $s$ and goal vertex $t$, the objective is to return one representative path for every distinct cost vector that is not strictly dominated by another $s$--$t$ path. Duplicate cost vectors are represented once.

\subsection{Labels and Consistent Lower Bounds}

A label $x$ represents a concrete $s$--$v(x)$ path and stores
\[
 x=\big(v(x),\gvec(x),\fvec(x),\operatorname{parent}(x),\nu_v(x),\nu_t(x)\big),
\]
where $\gvec(x)$ is the accumulated path cost, $\hvec(v)$ is a lower bound from $v$ to $t$, and
\[
    \fvec(x)=\gvec(x)+\hvec(v(x)).
\]
The parent reference reconstructs a path. The two version fields record the endpoint and goal snapshots observed when the label was generated.

The lower bounds are computed independently for each objective. For objective $i$, Dijkstra's algorithm is run from $t$ on the reverse graph using scalar edge cost $c_i$. The result $h_i(v)$ is the exact single-objective distance from $v$ to $t$ under objective $i$. Nonnegative edge costs imply consistency,
\[
    h_i(u)\le c_i(u,v)+h_i(v), \qquad (u,v)\in E,
\]
and therefore $\hvec$ is an admissible vector lower bound. If an existing goal vector $z$ satisfies $z\dom\fvec(x)$, every completion of $x$ is dominated and $x$ is safely pruned.

Each search worker uses one cyclic objective order $\pi_r$ and orders its local priority queue lexicographically by $\fvec$ under that order. Lexicographic label orders are standard in exact label-setting search, and parallel MOSP work has shown that several eligible labels can be processed concurrently when admission remains exact~\cite{Sanders2013,Ahmadi2025}. Different workers may therefore use different cyclic orders. Exactness does not require a globally unique minimum in \algname{} because labels are admitted only after exact endpoint-frontier validation and every admitted non-goal label is eventually expanded.

\subsection{Canonical Frontiers and Published Snapshots}

For each vertex $v$, let $\front_v$ denote its \emph{canonical frontier}: the exact mutable set of currently nondominated cost vectors accepted for $v$, together with representative labels. The term canonical distinguishes this set from read-only copies used for pruning. Each non-goal vertex is assigned to exactly one update worker, and only that worker modifies $\front_v$. The goal frontier $\front_t$ is modified only by the dedicated goal-update worker.

Search workers query a published snapshot rather than $\front_v$ directly. Snapshot version $k$ has the form
\[
    \snap_v^k=(\base_v^k,\deltaf_v^k,\idx_v^k,k).
\]
The \emph{base} $\base_v^k$ is a consolidated immutable array sorted in lexicographic order. The \emph{delta} $\deltaf_v^k$ is a small immutable list of vectors accepted after the base was last consolidated. Here, delta means recent frontier additions, not a numerical difference between cost vectors. The index $\idx_v^k$ summarizes only the base. Every current canonical vector appears in $\base_v^k\cup\deltaf_v^k$, although a snapshot may also contain older vectors that have since become dominated. Hence,
\[
    \front_v \subseteq \base_v^k\cup\deltaf_v^k.
\]
This conservative-superset property makes any positive snapshot dominance result safe.

\section{The \algname{} Framework}
\label{sec:framework}

\subsection{Asynchronous Search and Exclusive Frontier Updates}

Fig.~\ref{fig:architecture} summarizes the computation. The framework uses $R$ search workers, $P$ update workers for non-goal vertices, and one dedicated goal-update worker. The non-goal vertices are partitioned into $P$ ownership sets $V_1,\ldots,V_P$; update worker $p$ is responsible for every canonical frontier $\front_v$ with $v\in V_p$. Every worker has two local structures. A worker-specific incoming queue accepts concurrent label arrivals from other threads and has one consumer. The worker periodically transfers these arrivals into a private priority queue. Only the worker modifies that priority queue, so concurrent senders never perform heap insertions and cannot conflict on the ordered structure.

A search worker removes an accepted label from its best-first queue, expands all outgoing edges, and applies read-only snapshot pruning. If a candidate label $y$ ends at non-goal vertex $v$, it is appended directly to the incoming queue of update worker $\operatorname{owner}(v)$. If $v=t$, it is sent to the goal-update worker. The combination of a multi-producer incoming queue and a worker-private best-first queue enables this \emph{direct owner delivery}: many search workers may deliver labels concurrently, while only the receiving update worker restores local order and modifies $\front_v$. The centralized redistribution stage used by DAMPC is therefore unnecessary, and useful frontier information can be admitted and published with one less queueing stage.

When an update worker accepts a non-goal label, it selects a search worker using the power-of-two-choices strategy~\cite{Azar1994}. Two candidate search workers are obtained from deterministic hashes of the endpoint, their reported queue loads are read, and the label is sent to the less loaded one. This constant-size choice provides load balancing without scanning all search workers and preserves some endpoint locality. The framework does not depend on this particular policy; another dynamic load-balancing rule can be substituted without changing frontier ownership or correctness. The goal-update worker does not re-expand accepted goal labels; it stores one representative path for each accepted goal cost and publishes the updated goal snapshot.

\begin{figure*}[!t]
    \centering
    \includegraphics[width=0.98\textwidth]{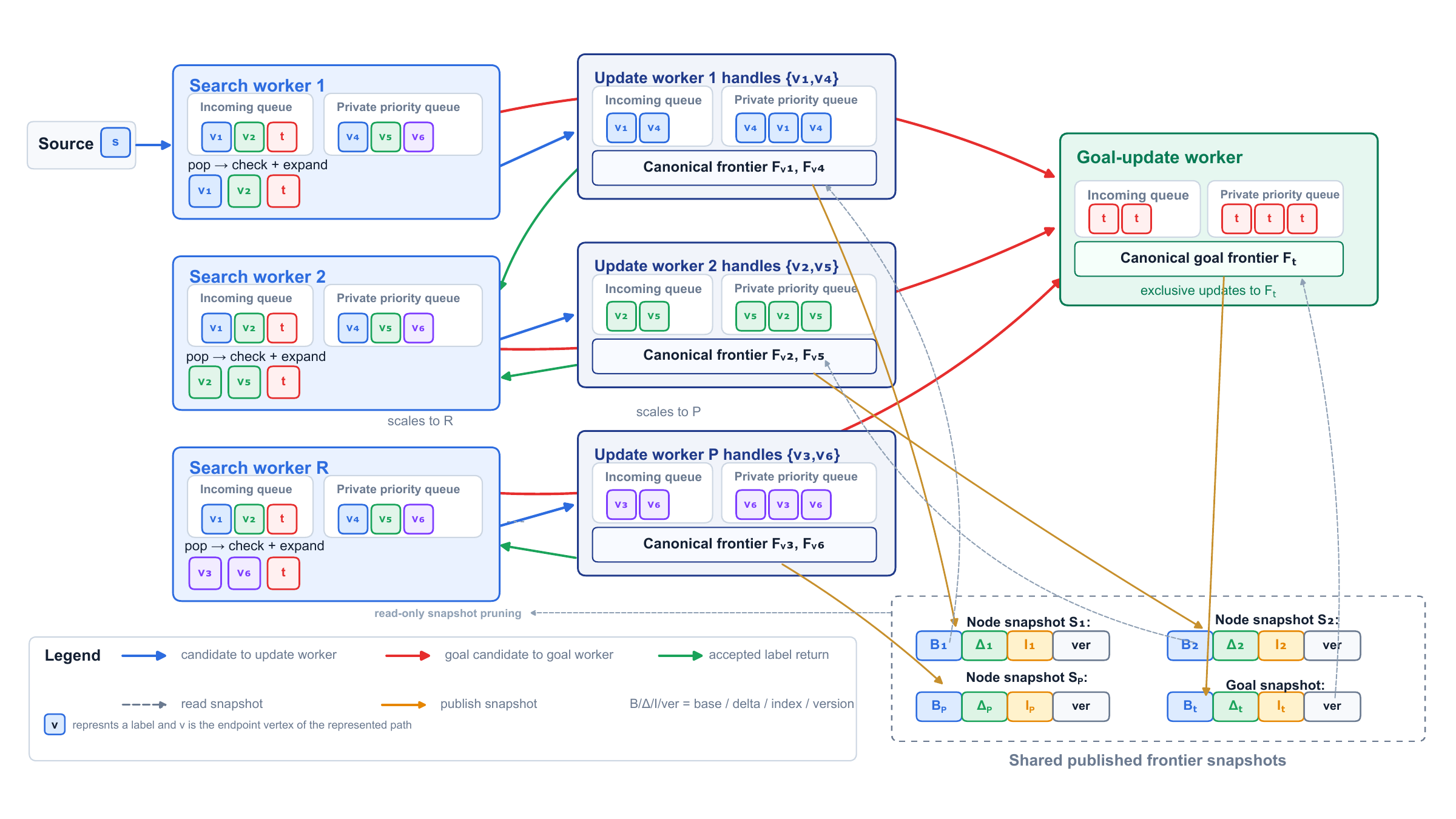}
    \caption{Architecture of \algname{}. Each small label box is identified by its endpoint; in particular, the symbol $v$ denotes the terminal vertex of the path represented by that label. A complete label also stores accumulated and lower-bound costs, a parent reference, and the snapshot versions observed during pruning. Each worker receives concurrent arrivals through a single-consumer incoming queue and alone modifies its private priority queue. Search workers send candidates directly to the update worker assigned to the endpoint. Solid arrows transfer labels, and dashed arrows denote immutable snapshot reads.}
    \label{fig:architecture}
\end{figure*}

\subsection{Versioned Base--Delta Snapshots}

Publishing a fully rebuilt indexed copy after every accepted vector would make snapshot construction proportional to the entire frontier size. \algname{} instead separates stable and recent state. If the frontier is small, the update worker publishes a consolidated base-only snapshot because copying and linear scanning are inexpensive. Once the frontier is large enough to index, the existing base and index can be shared by the next snapshot version, while the newly accepted vector is appended to the delta. The update worker consolidates again when the delta reaches an adaptive capacity or when the frontier has grown sufficiently relative to the base.

This organization serves two purposes. First, it keeps publication timely: a newly accepted vector enters the next immutable snapshot immediately and can prune labels already waiting in search queues. Second, it amortizes base copying and index construction over multiple accepted vectors. During the interval between consolidations, vectors dominated by newer additions may remain in the base or delta. Retaining them is safe. If an old vector $a$ dominates a candidate $q$, and a newer canonical vector $b$ dominates $a$, then $b\dom a\dom q$; the candidate is still safely rejected. The only effect of stale dominated vectors is extra query work.

The delta is queried before the base because recent vectors frequently invalidate queued or newly generated labels. Each delta stores a component-wise minimum vector $m_{\Delta}$. If $m_{\Delta,j}>q_j$ for any dimension $j$, no delta vector can dominate $q$, and the entire delta scan is skipped. If this inexpensive necessary condition does not rule out dominance, the delta is scanned exactly.

\subsection{Indexed Dominance Pruning}
\label{subsec:index}

After a delta miss, the query turns to the lexicographically ordered base. Let $q$ be the candidate vector and let
\[
 u=\max\{i: b_i\preceq_{\mathrm{lex}}q\},
\]
where $b_1,\ldots,b_n$ are the base vectors. The value $u$ is found by binary search. Every possible dominator lies in $b_1,\ldots,b_u$: a vector lexicographically larger than $q$ exceeds $q$ in its first differing dimension and therefore cannot weakly dominate it. Lexicographically ordered early termination is also exploited by NWMOA*~\cite{Ahmadi2024}; \algname{} applies the property to immutable base snapshots and adds range-minimum indexes over the feasible prefix.

For every indexed range $R$, the summary stores
\[
    m_j(R)=\min_{b\in R} b_j, \qquad j=2,\ldots,d.
\]
The first component need not be summarized because $b_1\le q_1$ already holds in the feasible lexicographic prefix. A range is impossible when $m_j(R)>q_j$ for at least one summarized dimension. The converse is not sufficient because the component minima may come from different vectors, so every vector in a surviving range is still checked exactly. The two indexes adapt standard blocked range summaries and segment-tree range decomposition~\cite{Bender2000,deBerg2008} to this necessary-condition test.

\paragraph{Block-minimum index.}
\algbm{} partitions the ordered base into fixed-size contiguous blocks of $b$ vectors. For each block and each truncated dimension $j=2,\ldots,d$, it stores the smallest component value appearing anywhere in that block. Query processing visits only blocks that intersect the feasible prefix. If one stored block minimum exceeds the corresponding component of $q$, every vector in that block fails dominance in that dimension and the whole block is skipped. Otherwise, the query performs exact dominance tests on the vectors of that block that lie inside the prefix. The index occupies $O((d-1)\lceil n/b\rceil)$ values and is traversed sequentially. Smaller blocks provide finer pruning but more summaries, whereas larger blocks reduce metadata and may require more exact comparisons.

\paragraph{Segment-tree-minimum index.}
\algst{} builds a complete binary segment tree over the ordered base. Every leaf represents one base vector, and every internal node represents the union of the consecutive intervals covered by its children. Each node stores the $d-1$ component minima of its interval. Query processing begins at the root, ignores nodes outside the feasible prefix, and prunes an entire represented interval when any stored minimum exceeds the corresponding query component. A surviving internal node is recursively refined into its children, and an exact dominance test is performed only at surviving leaves. Compared with fixed blocks, the tree can eliminate intervals at multiple granularities, including large prefixes near the root and smaller subranges deeper in the tree. This finer pruning requires $O((d-1)n)$ summary values and a less sequential traversal.

Fig.~\ref{fig:snapshot-query} gives a concrete three-dimensional example. The published snapshot contains the ordered base $b_1,\ldots,b_8$, an index over that base, and a recent delta. For $q=(7,6,6)$, the delta is checked first and is assumed not to contain a dominator. Binary search then excludes $b_7=(8,1,1)$ and $b_8=(9,0,0)$ because both are lexicographically larger than $q$, leaving six possible base vectors. Without an index, the algorithm must perform an exact dominance comparison with each of $b_1,\ldots,b_6$. With block size two, the first block has a summarized component minimum that exceeds $q$ in one dimension, and the second block fails in another dimension; both blocks are eliminated without inspecting their four vectors. Only $b_5$ and $b_6$ are compared exactly, and $b_6=(7,5,5)$ dominates $q$. The segment tree reaches the same result by eliminating interval $[1,4]$ at an internal node, descending into interval $[5,6]$, and ignoring $[7,8]$ because it lies outside the feasible prefix. The example therefore reduces six exact base comparisons to two. On large Pareto frontiers, a single block or tree-node test can exclude tens, hundreds, or thousands of vectors, so the avoided exact comparisons accumulate over the many dominance queries performed by the search. Neither index changes the final dominance predicate.

\begin{figure*}[!t]
    \centering
    \includegraphics[width=0.98\textwidth]{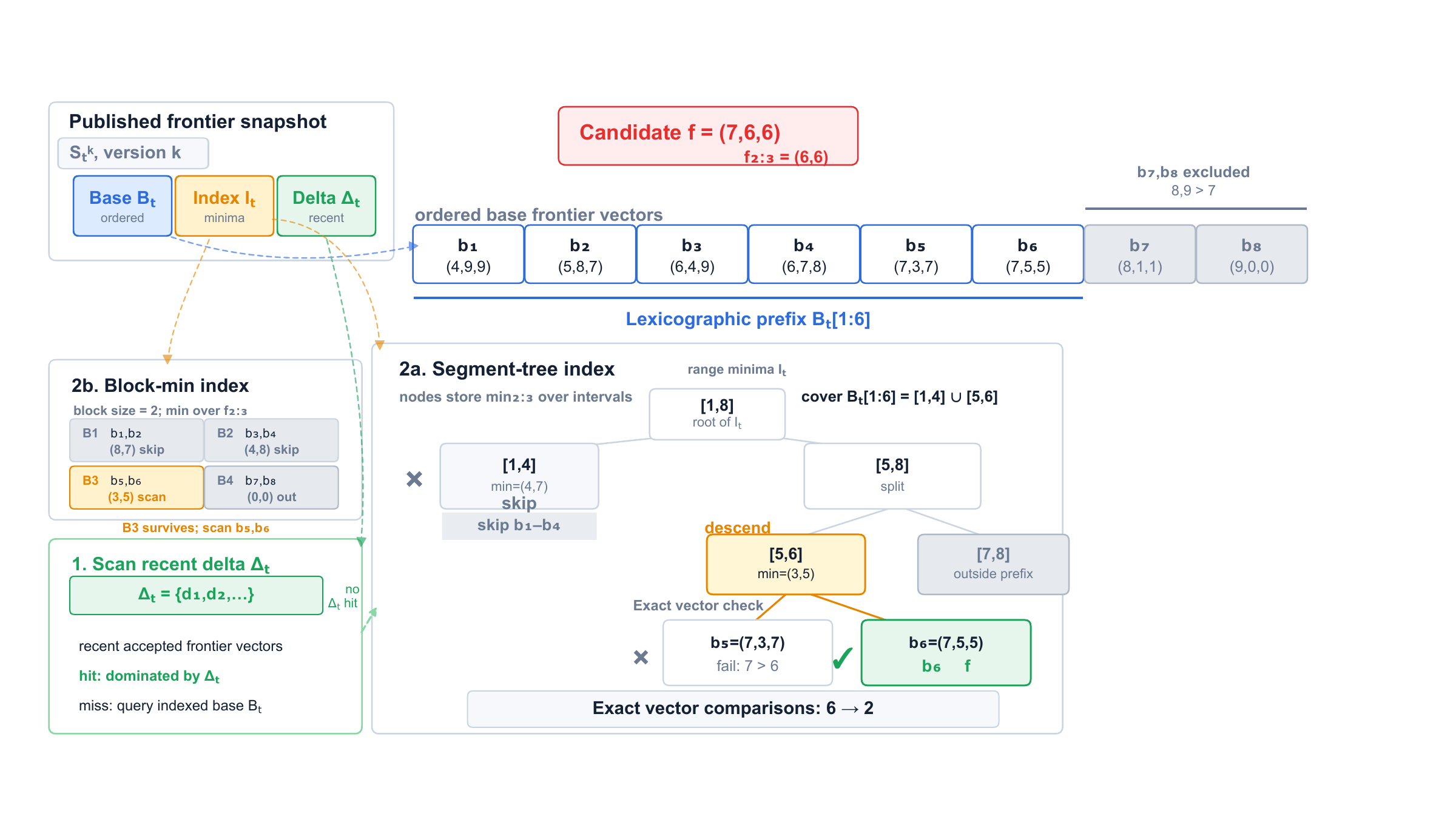}
    \caption{Dominance query over a published versioned frontier snapshot. The snapshot contains a recent delta and an ordered consolidated base; the base is queried through either a block-minimum index or a segment-tree-minimum index. In the example, the lexicographic prefix removes two vectors and either index reduces the remaining six exact base comparisons to two.}
    \label{fig:snapshot-query}
\end{figure*}

\subsection{Version Reuse}
\label{subsec:version}

A \emph{negative dominance result} means that a snapshot query returns false: no vector represented by that immutable snapshot weakly dominates the candidate. The search worker therefore retains and forwards the candidate rather than pruning it. However, this result may become stale before the candidate is processed by its update worker if the corresponding frontier changes in the meantime.The update worker reuses the negative result only if the observed snapshot version is still current; otherwise, it repeats the dominance query against the latest snapshot.

\algname{} attaches the observed version $\nu$ to every surviving candidate. When the assigned update worker receives the candidate, it compares $\nu$ with the current snapshot version. If they are equal, the queried immutable set is unchanged, so the negative result remains valid and the update worker need not traverse the delta and base a second time. If the version differs, the update worker queries the current snapshot. A positive result rejects the candidate; a negative result admits it to the canonical-frontier update. The same rule is applied by the goal-update worker to complete labels.

Version reuse shortens the common path in which a candidate label moves quickly from a search worker to its endpoint update worker and no intervening frontier update occurs. It is especially useful when snapshot queries are expensive, frontiers are large, or update-queue delay is short enough that the observed version often remains current. The implementation records the version examined by the search worker in the candidate label. At the update worker, a version mismatch causes the dominance test to be repeated against the current snapshot before the canonical frontier is modified. Thus, frequent update traffic does not weaken correctness; it simply reduces the reuse rate and causes more exact revalidations. Fig.~\ref{fig:version-reuse} illustrates the two cases.

\begin{figure*}[!t]
    \centering
    \includegraphics[width=0.93\textwidth]{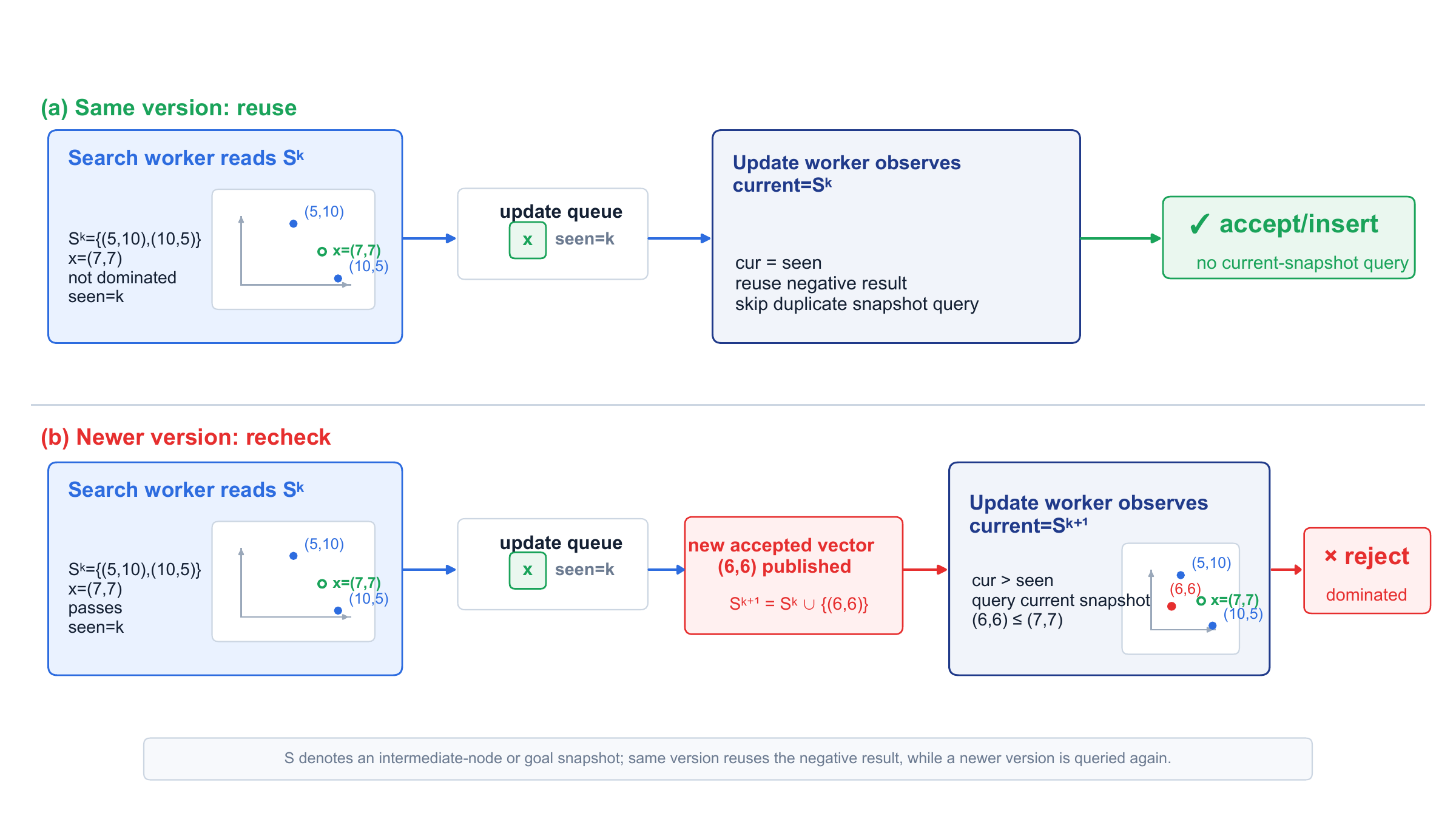}
    \caption{Version reuse after a negative dominance result. An unchanged snapshot version certifies that the search worker and the update worker refer to the same immutable set, so the duplicate query is skipped. A changed version triggers an exact query against the current snapshot.}
    \label{fig:version-reuse}
\end{figure*}

\subsection{Core Algorithms}

Algorithms~\ref{alg:main}--\ref{alg:index} connect the preceding components into the complete search. Algorithm~\ref{alg:main} performs the one-time initialization: it computes one reverse single-objective distance function per objective, assigns every non-goal vertex to one update worker, publishes empty snapshots, submits the start label, launches the two worker loops, and waits for global completion. The incoming queues in the pseudocode support concurrent append and one consumer, whereas each local priority queue is modified only by its worker.

\begin{algorithm}[!t]
\caption{High-Level \algname{} Search}
\label{alg:main}
\begin{algorithmic}[1]
\Require $G=(V,E)$, costs $\cost$, start $s$, goal $t$, $R$ search workers, $P$ non-goal update workers
\Ensure One representative path for every cost-unique Pareto-optimal $s$--$t$ vector
\For{$i\gets1$ to $d$}
    \State Run scalar Dijkstra from $t$ on the reverse graph under $c_i$
    \State Set $h_i(v)$ to the resulting distance for every $v\in V$
\EndFor
\ForAll{$v\in V\setminus\{t\}$}
    \State Assign $v$ to exactly one update worker $\operatorname{owner}(v)$
    \State $\front_v\gets\emptyset$; publish $\snap_v^0=(\emptyset,\emptyset,\emptyset,0)$
\EndFor
\State Set $\operatorname{owner}(t)$ to the dedicated goal-update worker
\State $\front_t\gets\emptyset$; publish $\snap_t^0=(\emptyset,\emptyset,\emptyset,0)$
\State Create the search workers, update workers, and the dedicated goal-update worker
\State Create $x_s=(s,\mathbf{0},\hvec(s),\nil,0,0)$
\State Append $x_s$ to the incoming queue of $\operatorname{owner}(s)$
\State Execute Algorithms~\ref{alg:search} and~\ref{alg:update} concurrently
\State Wait until every incoming and local priority queue is empty and every worker is idle
\State Stop all workers and return the representatives stored with $\front_t$
\end{algorithmic}
\end{algorithm}

Algorithm~\ref{alg:search} describes the expansion side. A search worker first merges arrivals into its private best-first queue, extracts the locally smallest label under its assigned objective order, and expands all outgoing edges. Each generated candidate is pruned against immutable endpoint and goal snapshots before it is sent directly to the update worker responsible for its endpoint.

\begin{algorithm}[!t]
\caption{Search Worker $r$}
\label{alg:search}
\begin{algorithmic}[1]
\State Let $I_r$ be the incoming queue and $Q_r$ the private best-first queue ordered by $\pi_r$
\While{the search has not terminated}
    \State Transfer every available label from $I_r$ to $Q_r$
    \If{$Q_r=\emptyset$}
        \State Wait for an arrival or termination signal; \textbf{continue}
    \EndIf
    \State Remove $x\in Q_r$ with minimum lexicographic $\fvec(x)$ under order $\pi_r$
    \State Read the current immutable goal snapshot $\snap_t^{\nu_t}$
    \ForAll{$(v(x),u)\in E$}
        \State $\gvec(y)\gets\gvec(x)+\cost(v(x),u)$; $\fvec(y)\gets\gvec(y)+\hvec(u)$
        \State Read the current immutable endpoint snapshot $\snap_u^{\nu_u}$
        \If{\textsc{SnapshotDominates}$(\snap_u^{\nu_u},\gvec(y))$}
            \State \textbf{continue}
        \EndIf
        \If{\textsc{SnapshotDominates}$(\snap_t^{\nu_t},\fvec(y))$}
            \State \textbf{continue}
        \EndIf
        \State Set $y=(u,\gvec(y),\fvec(y),x,\nu_u,\nu_t)$
        \State Append $y$ to the incoming queue of $\operatorname{owner}(u)$
    \EndFor
\EndWhile
\end{algorithmic}
\end{algorithm}

Algorithm~\ref{alg:update} gives the update side. For each candidate, the receiving worker first determines whether the search-side negative result can be reused from the carried snapshot version. If the version has changed, the current snapshot is queried again. A surviving candidate is then inserted into the exact canonical frontier, the next base--delta snapshot is published, and a non-goal label is returned to a lightly loaded search worker. The same procedure applies to the dedicated goal-update worker, except that accepted goal labels are retained as solution representatives and are not expanded again.

\begin{algorithm}[!t]
\caption{Update Worker for an Assigned Vertex Set}
\label{alg:update}
\begin{algorithmic}[1]
\State Let $J_p$ be the incoming queue and $U_p$ the private priority queue
\While{the search has not terminated}
    \State Transfer every available candidate from $J_p$ to $U_p$
    \If{$U_p=\emptyset$}
        \State Wait for an arrival or termination signal; \textbf{continue}
    \EndIf
    \State Remove a candidate $y\in U_p$ with minimum lexicographic $\fvec(y)$; let $v=v(y)$
    \State Read current snapshot $\snap_v^{\nu}$
    \State $\widehat{\nu}(y)\gets\nu_t(y)$ if $v=t$, and $\widehat{\nu}(y)\gets\nu_v(y)$ otherwise
    \If{$\widehat{\nu}(y)\ne\nu$ \textbf{and} \textsc{SnapshotDominates}$(\snap_v^{\nu},\gvec(y))$}
        \State Discard $y$; \textbf{continue}
    \EndIf
    \State $\front_v\gets\big(\front_v\setminus\{z:\gvec(y)\dom z\}\big)\cup\{\gvec(y)\}$
    \If{$\lvert\front_v\rvert<\theta_I$}
        \State $\base_v\gets$ lexicographically sorted copy of
           $\front_v$
        \State $\deltaf_v\gets\emptyset$; $\idx_v\gets\emptyset$
    \ElsIf{consolidation is due}
    \State $\base_v\gets$ lexicographically sorted copy of
           $\front_v$
    \State $\deltaf_v\gets\emptyset$
        \State Build the selected block-minimum or segment-tree-minimum index $\idx_v$ over $\base_v$
    \Else
        \State Reuse $\base_v$ and $\idx_v$; set $\deltaf_v\gets\deltaf_v\cup\{\gvec(y)\}$
    \EndIf
    \State Publish $\snap_v^{\nu+1}=(\base_v,\deltaf_v,\idx_v,\nu+1)$
    \If{$v=t$}
        \State Store $y$ as the representative of cost $\gvec(y)$
    \Else
        \State Obtain two search-worker candidates $r_1,r_2$ from hashes of $v$
        \State $r\gets r_1$ if $\operatorname{load}(r_1)\le\operatorname{load}(r_2)$, else $r\gets r_2$
        \State Append accepted label $y$ to incoming queue $I_r$
    \EndIf
\EndWhile
\end{algorithmic}
\end{algorithm}

Algorithms~\ref{alg:dom} and~\ref{alg:index} specify the exact read-only query used by both worker types. The query checks the recent delta first, restricts the ordered base to the lexicographically feasible prefix, and then applies either exact linear scanning or one of the two range-minimum indexes. The index threshold $\theta_I$ and adaptive delta capacity $\theta_\Delta(\lvert\front_v\rvert)$ affect only when a snapshot is consolidated; they do not alter the dominance predicate.

\begin{algorithm}[!t]
\caption{\textsc{SnapshotDominates}$(\snap^\nu,q)$}
\label{alg:dom}
\begin{algorithmic}[1]
\Require $\snap^\nu=(\base,\deltaf,\idx,\nu)$ and query vector $q$
\If{$\deltaf\ne\emptyset$}
    \State $m_j\gets\min_{z\in\deltaf}z_j$ for $j=1,\ldots,d$
    \If{$m_j\le q_j$ for every $j$ \textbf{and} some $z\in\deltaf$ satisfies $z\dom q$}
        \State \Return \textbf{true}
    \EndIf
\EndIf
\State By binary search, find the largest $u$ with $b_u\preceq_{\mathrm{lex}}q$
\If{$u=0$}
    \State \Return \textbf{false}
\EndIf
\If{$\idx=\emptyset$}
    \State \Return whether some $b_i\in\base[1{:}u]$ satisfies $b_i\dom q$
\Else
    \State \Return \textsc{IndexedPrefixDominates}$(\idx,\base,q,u)$
\EndIf
\end{algorithmic}
\end{algorithm}

\begin{algorithm}[!t]
\caption{\textsc{IndexedPrefixDominates}$(\idx,\base,q,u)$}
\label{alg:index}
\begin{algorithmic}[1]
\If{$\idx$ is a block-minimum index}
    \ForAll{blocks $C$ intersecting $\base[1{:}u]$}
        \If{$\min_{z\in C}z_j\le q_j$ for every $j=2,\ldots,d$}
            \If{some $z\in C\cap\base[1{:}u]$ satisfies $z\dom q$} \State \Return \textbf{true} \EndIf
        \EndIf
    \EndFor
\Else
    \State Initialize a stack with the segment-tree root and interval $[1,\lvert\base\rvert]$
    \While{the stack is not empty}
        \State Remove an indexed interval $C$; skip it if $C\cap[1,u]=\emptyset$
        \If{$\min_{z\in C}z_j>q_j$ for some $j=2,\ldots,d$} \State \textbf{continue} \EndIf
        \If{$C$ is a leaf and its vector dominates $q$} \State \Return \textbf{true} \EndIf
        \If{$C$ is not a leaf} \State Push its two child intervals \EndIf
    \EndWhile
\EndIf
\State \Return \textbf{false}
\end{algorithmic}
\end{algorithm}

Together, Algorithm~\ref{alg:main} supplies global control, Algorithms~\ref{alg:search} and~\ref{alg:update} form the asynchronous pipeline, and Algorithms~\ref{alg:dom} and~\ref{alg:index} implement exact snapshot pruning.

\subsection{Termination}

The implementation terminates on stable global quiescence rather than updating one global counter for every label transfer. A worker declares itself idle only when it has no current label, its private priority queue is empty, and its incoming queue has no published or in-progress item. The coordinator confirms that all worker states and queue cursors remain unchanged across repeated observations before issuing the stop signal. Since the system is closed after the start label is submitted, a stable state in which every worker is idle and every queue is empty cannot generate future labels. This protocol is an implementation of the abstract termination condition in Algorithm~\ref{alg:main}; it does not affect search ordering or dominance decisions.

\section{Correctness and Complexity}
\label{sec:correctness}

\begin{lemma}[Canonical-frontier invariant]
After every update-worker operation, $\front_v$ contains exactly the cost-unique nondominated vectors among all labels admitted so far for vertex $v$.
\end{lemma}
\begin{proof}
A candidate is admitted only if the current snapshot contains no weak dominator. The snapshot is a superset of the current canonical frontier, so no vector in $\front_v$ dominates or duplicates the candidate. The update removes every stored vector dominated by the candidate and inserts the candidate once. No incomparable vector is removed. The invariant follows by induction over the serial update sequence for $v$.
\end{proof}

\begin{lemma}[Snapshot safety]
If \textsc{SnapshotDominates}$(\snap_v^k,q)$ returns true, then $q$ is dominated by a feasible path to $v$, even if the snapshot contains a vector that is no longer in $\front_v$.
\end{lemma}
\begin{proof}
Every snapshot vector was produced by an admitted feasible label. If a snapshot vector $a$ has since been removed, it was removed because a later feasible vector $b$ satisfies $b\dom a$. If $a\dom q$, transitivity gives $b\dom q$. Thus, a positive result remains safe after later frontier updates.
\end{proof}

\begin{lemma}[Exact indexed query]
Algorithms~\ref{alg:dom} and~\ref{alg:index} return true if and only if some vector represented by the queried snapshot weakly dominates $q$.
\end{lemma}
\begin{proof}
The delta is tested by exact dominance. Any base dominator must be lexicographically no larger than $q$, so it lies in the prefix $[1,u]$. A block or tree interval is skipped only if a component minimum exceeds $q$; then every vector in the interval exceeds $q$ in that component. Surviving vectors are checked exactly. Consequently, no dominator is skipped and no non-dominator can cause a positive result.
\end{proof}

\begin{lemma}[Version-reuse safety]
Suppose a search worker obtains \textsc{SnapshotDominates}$(\snap_v^k,q)=\mathrm{false}$ and the update worker later observes the same version $k$. Then no vector in the current canonical frontier $\front_v$ weakly dominates $q$.
\end{lemma}
\begin{proof}
The version equality identifies the same immutable snapshot. By the exact-query lemma, no represented vector dominates $q$. Since $\front_v\subseteq\base_v^k\cup\deltaf_v^k$, no canonical vector dominates $q$ either.
\end{proof}

\begin{lemma}[Goal pruning]
If a goal vector $z\in\front_t$ satisfies $z\dom\fvec(x)$, then no completion of label $x$ can produce a new cost-unique Pareto-optimal goal vector.
\end{lemma}
\begin{proof}
Admissibility gives $\fvec(x)\dom\cost(\pi)$ in the minimization sense that every feasible completion cost is component-wise at least $\fvec(x)$. Therefore $z\dom\fvec(x)\dom\cost(\pi)$ for every completion $\pi$ of $x$.
\end{proof}

\begin{theorem}[Exactness]
Suppose every queued label is eventually processed unless it is safely pruned, and termination is declared only at stable global quiescence. If \algname{} terminates without a timeout, the goal frontier $\front_t$ contains exactly one path representative for every Pareto-optimal $s$--$t$ cost vector.
\end{theorem}
\begin{proof}
Every label is generated by extending the initial label along graph edges and therefore represents a feasible path. The goal-update worker applies an exact Pareto insertion operation: it rejects a goal cost weakly dominated by an existing one, removes existing costs strictly dominated by a newly admitted one, and retains one representative for duplicate cost vectors. Hence $\front_t$ is always cost-unique and nondominated among the admitted goal labels.

For completeness, fix a Pareto-optimal path $P=\langle v_0=s,v_1,\ldots,v_m=t\rangle$ and let $\gvec_i$ be the cost of its prefix ending at $v_i$. We show by induction that, for every $i$, the algorithm admits a label at $v_i$ with cost $\gvec_i$, or already contains an equivalent-cost representative. The claim holds for the initial label at $s$.

Assume it holds at $v_i$. Fair processing eventually expands the corresponding label and generates a candidate for $v_{i+1}$ with cost $\gvec_{i+1}$. If an endpoint-frontier query rejects this candidate, some feasible label at $v_{i+1}$ weakly dominates $\gvec_{i+1}$. Strict dominance would, after appending the remaining suffix of $P$, produce an $s$--$t$ path that strictly dominates $P$, contradicting Pareto optimality. Equality merely provides the required equivalent-cost representative.

The candidate also cannot be incorrectly removed by goal pruning. Admissibility gives $\fvec_{i+1}\dom\cost(P)$. Thus, any admitted goal cost that dominates $\fvec_{i+1}$ also weakly dominates $\cost(P)$. This either contradicts Pareto optimality or shows that the same goal cost is already represented in $\front_t$.

A positive snapshot result is safe because every vector in a published snapshot originates from a feasible admitted label. A negative result is reused only when the snapshot version observed by the search worker is still current; otherwise, the update worker repeats the dominance query against the latest snapshot. Therefore the candidate, or an equivalent-cost representative, is admitted. The induction reaches $v_m=t$, so every Pareto-optimal cost vector is admitted to $\front_t$.

Finally, suppose a globally dominated cost remained in $\front_t$ at termination. It is dominated by some Pareto-optimal $s$--$t$ cost, which has been admitted by the completeness argument; the exact goal-frontier update would therefore have rejected or removed the dominated cost, a contradiction. Stable global quiescence guarantees that no queued, active, or in-flight candidate remains unprocessed. Consequently, $\front_t$ contains exactly one representative for every Pareto-optimal $s$--$t$ cost vector.
\end{proof}

\subsection{Complexity}

Let $n=\lvert\base\rvert$, $\delta=\lvert\deltaf\rvert$, $u$ be the feasible prefix length, and $d$ the objective count. The delta test is $O(d\delta)$ in the worst case, while the lexicographic binary search is $O(d\log n)$. With no index, the base query is $O(du)$. For block size $b$, \algbm{} tests $\lceil u/b\rceil$ summaries and exactly scans only the blocks that survive; if $s$ vectors lie in surviving blocks, the cost is $O((d-1)\lceil u/b\rceil+ds)$. \algst{} visits only segment-tree nodes that intersect the prefix and are not eliminated by interval minima. If $z$ index nodes and $\ell$ leaves are visited, the cost is $O((d-1)z+d\ell)$. In arbitrary dimension, both indexes retain $O(du)$ worst-case query time because summary minima from different dimensions may be attained by different vectors. Their benefit is data-dependent range elimination, not a logarithmic worst-case bound.

A consolidation copies and sorts the canonical frontier and builds an $O(dn)$ index. The base--delta policy avoids performing this work after every insertion. Between consolidations, publication copies only the small delta and shares the immutable base. Endpoint ownership serializes updates to a single vertex but permits independent vertices to be updated in parallel. The search itself remains output sensitive: no exact method can avoid costs proportional to a potentially exponential Pareto frontier.

\section{Experimental Evaluation}
\label{sec:evaluation}

\subsection{Experimental Setup}

Experiments ran on a server with a 56-core Intel Xeon Gold processor at 2.60~GHz, 112 hardware threads, 512~GB RAM, and Ubuntu 22.04.5 LTS. All C++ programs were compiled with the GNU \texttt{g++} driver using \texttt{-O3}. Every run retained parent references for accepted labels so that complete solution paths could be reconstructed. We report main-search time and use the operating-system peak resident set size as the memory measure; below, this quantity is referred to simply as memory. The timeout was one hour per query.

We compare SIP-MOSP with four state-of-the-art exact MOSP baselines: LTMOA*, NWMOA*, Parallel LTMOA*, and Parallel NWMOA*. All four use the source code released with Parallelizing Multi-objective A* Search~\cite{Ahmadi2025}.\footnote{\url{https://bitbucket.org/s-ahmadi/multiobj/src/main/}} The released code reports an internal memory estimate rather than operating-system peak memory. We added the same peak-memory measurement to all implementations and otherwise left the baseline algorithms unchanged. The parallel baselines run one search per cyclic objective order and therefore use $d$ search threads. The principal \algname{} experiments use 96 workers, divided into search/update roles as 40/56 on NYC-Road, 72/24 on AS3356, and 88/8 on Dense-180; the dedicated goal updater is included in the update count. These splits were selected empirically from short pilot runs on representative queries and then fixed for every reported query of the corresponding topology. Search workers use four cyclic order groups on NYC-Road and ten groups on AS3356 and Dense-180.

Each setting uses 40 start--goal pairs, except for the five-query Dense-180 ablation subset. For each main comparison, mean time and memory are computed over the mutually solved instances, that is, the queries completed by every algorithm included in that comparison. This avoids treating a timeout as a one-hour numeric runtime or averaging different methods over query subsets of unequal difficulty. All speedup factors are computed from the unrounded mean solve times over the same mutually solved instances; the values shown in the tables are rounded only for presentation. The ``Solved'' column reports the total number of queries completed by each method. For Dense-180 with 20 objectives, the mean comparison includes only the four parallel methods because the sequential methods exceed the one-hour limit on most instances.

\begin{table}[!t]
\centering
\caption{Evaluation topologies.}
\label{tab:topologies}
\footnotesize
\setlength{\tabcolsep}{3.4pt}
\begin{tabular}{lrrl}
\toprule
Topology & Vertices & Directed arcs & Objectives \\
\midrule
NYC-Road & 264,346 & 733,846 & 3, 4 \\
AS3356 & 624 & 10,596 & 10, 20 \\
Dense-180 & 180 & 32,220 & 10, 20 \\
\bottomrule
\end{tabular}
\end{table}

NYC-Road is the New York road network from the 9th DIMACS Implementation Challenge~\cite{DIMACS2009}. We use the benchmark cost components defined by Peer \emph{et al.}~\cite{Peer2026}. The first two objectives are distance and travel time. The third objective in our experiments is the mean out-degree of the edge endpoints, which reflects how strongly an edge is connected to surrounding roads. The fourth objective assigns unit cost to every edge, so a path's fourth component counts the number of traversed road segments. Thus, the three-objective setting uses distance, travel time, and endpoint connectivity, while the four-objective setting additionally counts traversed links. We use the same first 40 start--goal pairs distributed with the benchmark suite.\footnote{\url{https://github.com/CRL-Technion/Multi-Objective-Search-Benchmarks/tree/main/benchmarks/road/DIMACS}}

AS3356 is a Rocketfuel ISP topology~\cite{Spring2002}. Dense-180 is the complete directed graph used to stress frontier growth under high density. For AS3356 and Dense-180, each objective component is an independently assigned integer edge cost in $[0,100]$, following the supplied benchmark construction. The same graph, objective count, and query pairs are used for all algorithms in a setting.

\subsection{Main Results}

Tables~\ref{tab:nyc-main}--\ref{tab:dense-main} report solve time and memory together for each topology. A value in the ``Solved'' column is the number of queries completed within one hour. For each objective setting, mean time and memory are computed only over the queries completed by every algorithm included in that column group.

\begin{table*}[!t]
\centering
\caption{NYC-Road results.}
\label{tab:nyc-main}
\scriptsize
\setlength{\tabcolsep}{3.6pt}
\begin{tabular}{lrrr|rrr}
\toprule
& \multicolumn{3}{c|}{3 objectives (avg. $|\mathrm{Sols}|=7{,}280$)} & \multicolumn{3}{c}{4 objectives (avg. $|\mathrm{Sols}|=16{,}057$)} \\
Method & Solved & Time (s) & Mem. (GB) & Solved & Time (s) & Mem. (GB) \\
\midrule
LTMOA*       & 40 & 41.35 & 0.251 & 39 & 145.77 & 0.336 \\
Parallel LTMOA* & 40 & 10.05 & 0.442 & 40 & 39.83 & 0.835 \\
NWMOA*       & 40 & 9.88 & \textbf{0.244} & 40 & 40.58 & \textbf{0.320} \\
Parallel NWMOA* & 40 & 6.33 & 0.408 & 40 & 20.29 & 0.737 \\
\algbm{}     & 40 & \textbf{4.78} & 0.865 & 40 & \textbf{19.57} & 1.582 \\
\algst{}     & 40 & 6.09 & 1.262 & 40 & 26.53 & 2.559 \\
\bottomrule
\end{tabular}
\end{table*}

\begin{table*}[!t]
\centering
\caption{AS3356 results.}
\label{tab:as-main}
\scriptsize
\setlength{\tabcolsep}{3.6pt}
\begin{tabular}{lrrr|rrr}
\toprule
& \multicolumn{3}{c|}{10 objectives (avg. $|\mathrm{Sols}|=3{,}915$)} & \multicolumn{3}{c}{20 objectives (avg. $|\mathrm{Sols}|=26{,}108$)} \\
Method & Solved & Time (s) & Mem. (GB) & Solved & Time (s) & Mem. (GB) \\
\midrule
LTMOA*       & 40 & 11.36 & 0.116 & 38 & 382.15 & 1.114 \\
Parallel LTMOA* & 40 & 1.47 & 0.659 & 40 & 33.05 & 11.049 \\
NWMOA*       & 40 & 5.98 & \textbf{0.100} & 38 & 225.48 & 1.028 \\
Parallel NWMOA* & 40 & 1.51 & 0.616 & 40 & 20.74 & 10.738 \\
\algbm{}     & 40 & 0.29 & 0.104 & 40 & 8.20 & \textbf{0.503} \\
\algst{}     & 40 & \textbf{0.26} & 0.146 & 40 & \textbf{4.81} & 0.987 \\
\bottomrule
\end{tabular}
\end{table*}

\begin{table*}[!t]
\centering
\caption{Dense-180 results. Sequential methods are omitted from the 20-objective means because most queries exceed one hour.}
\label{tab:dense-main}
\scriptsize
\setlength{\tabcolsep}{3.6pt}
\begin{tabular}{lrrr|rrr}
\toprule
& \multicolumn{3}{c|}{10 objectives (avg. $|\mathrm{Sols}|=8{,}453$)} & \multicolumn{3}{c}{20 objectives (avg. $|\mathrm{Sols}|=96{,}035$)} \\
Method & Solved & Time (s) & Mem. (GB) & Solved & Time (s) & Mem. (GB) \\
\midrule
LTMOA*       & 40 & 114.55 & 1.150 & -- & -- & -- \\
Parallel LTMOA* & 40 & 8.22 & 6.633 & 40 & 229.63 & 167.231 \\
NWMOA*       & 40 & 40.07 & 0.969 & -- & -- & -- \\
Parallel NWMOA* & 40 & 9.18 & 5.986 & 40 & 183.19 & 163.244 \\
\algbm{}     & 40 & \textbf{1.17} & \textbf{0.117} & 40 & 100.16 & \textbf{1.421} \\
\algst{}     & 40 & 1.18 & 0.182 & 40 & \textbf{54.82} & 2.705 \\
\bottomrule
\end{tabular}
\end{table*}

Benchmark diversity is important because graph structure and objective construction can change branching, frontier growth, and the relative cost of search and dominance processing~\cite{Peer2026}. Across the three topology families, the same SIP-MOSP architecture remains effective rather than relying on a road-network-specific behavior. The higher-dimensional AS3356 and Dense-180 results are especially notable: the proposed methods improve both solve time and memory, showing that parallelism is obtained without replicating complete objective-order search states.

Both \algbm{} and \algst{} outperform LTMOA* and NWMOA* on the mutually solved instances in every setting for which a sequential comparison is available. NWMOA* has the lower mean solve time of the two sequential baselines in each such setting. We therefore use NWMOA* as the more conservative sequential reference below; the corresponding speedups over LTMOA* are larger.

In the three- and four-objective NYC-Road settings, \algbm{} achieves the lowest mean solve time. With three objectives, its mean time is 4.78~s, compared with 41.35~s for LTMOA*, 9.88~s for NWMOA*, 10.05~s for Parallel LTMOA*, and 6.33~s for Parallel NWMOA*. These results correspond to speedups of 8.65$\times$, 2.07$\times$, 2.10$\times$, and 1.32$\times$, respectively. With four objectives, \algbm{} requires 19.57~s, whereas the same four baselines require 145.77, 40.58, 39.83, and 20.29~s, yielding speedups of 7.45$\times$, 2.07$\times$, 2.04$\times$, and 1.04$\times$, respectively. Thus, although the advantage over Parallel NWMOA*, the best-performing parallel baseline, is narrow in the four-objective setting, the improvements over both sequential baselines and Parallel LTMOA* remain substantial. The sequential and objective-order parallel baselines use less memory in these low-dimensional road settings. \algbm{}, the better-performing \algname{} backend on NYC-Road, uses 0.865 and 1.582~GB in the three- and four-objective settings, respectively, whereas the four baselines use 0.244--0.442 and 0.320--0.835~GB. The difference is consistent with the larger worker pool, the additional concurrently maintained search state, and the potentially longer start--goal paths in this much larger-vertex road network. Thus, the memory benefit of sharing frontier state is less pronounced in these low-dimensional settings, but becomes substantial in the higher-dimensional workloads, where replicating large Pareto frontiers is considerably more costly.

On AS3356, SIP-MOSP shows a much larger runtime advantage and substantial memory savings relative to the parallel baselines. On the mutually solved instances, \algst{} is 23.0$\times$ and 46.9$\times$ faster than NWMOA* with ten and twenty objectives, respectively. Relative to Parallel LTMOA* in the 10-objective setting and Parallel NWMOA* in the 20-objective setting, the corresponding speedups are 5.65$\times$ and 4.31$\times$. In the 20-objective setting, Parallel NWMOA* uses 10.738~GB of memory, 10.9$\times$ the 0.987~GB used by \algst{}. \algbm{} uses only 0.503~GB, making the block-minimum backend the more memory-efficient alternative in this setting, although the segment-tree backend is faster. 

Dense-180 further amplifies both search effort and memory replication. With ten objectives, \algbm{} is 34.2$\times$ faster than NWMOA*, the best-performing sequential baseline, and 7.05$\times$ faster than Parallel LTMOA*, the best-performing parallel baseline, on the mutually solved instances. Parallel LTMOA* uses 6.633~GB of memory, 56.7$\times$ the 0.117~GB used by \algbm{}. With twenty objectives, the sequential methods are excluded from the mean-time comparison because most instances remain unsolved by them within the one-hour limit. \algst{} is 3.34$\times$ faster than Parallel NWMOA*. Parallel NWMOA* uses 163.244~GB of memory, 60.3$\times$ the 2.705~GB used by \algst{}. A key structural difference is that the objective-order parallel methods maintain multiple complete search states, whereas \algname{} maintains one canonical frontier per vertex, including the goal, within a single cooperative search. This difference is consistent with the large memory gaps observed in the high-dimensional settings.

\subsection{Component Ablations}
\label{subsec:ablation}

We evaluate five modifications for each index backend: one objective-order group, a centralized dispatcher instead of direct owner delivery, eager monolithic indexed snapshots instead of base--delta snapshots, no base index, and no version reuse. In the eager monolithic variant, every accepted label causes the complete canonical frontier to be copied into one consolidated snapshot and the applicable index to be rebuilt immediately; the delta is always empty. The full method is also shown. Each row is averaged over the queries completed by that variant. Hence a variant with no timeout uses all 40 road or AS3356 queries, whereas a variant with one timeout uses the remaining 39. The relative factor compares the variant with the full method on exactly that variant's solved subset, so unrelated variants do not discard queries because another variant timed out. Dense-180 twenty-objective runtimes are comparatively uniform, but each ablation run is long; we therefore use five random pairs for this component study.

\begin{table*}[!t]
\centering
\caption{NYC-Road, four objectives: component results for both index backends.}
\label{tab:ablation-ny}
\scriptsize
\begin{minipage}{0.49\textwidth}
\centering
\textbf{(a) \algbm{}}\\[2pt]
\begin{tabular}{lrrrr}
\toprule
Variant & $N$ & Time (s) & Rel. & Memory (GB) \\
\midrule
Full & 40 & 48.23 & 1.00 & 2.881 \\
One order group & 40 & 106.91 & 2.22 & 2.567 \\
Central dispatcher & 40 & 108.08 & 2.24 & 7.378 \\
Eager monolithic snapshot & 39 & 88.15 & 4.50 & 3.237 \\
No block-minimum index & 40 & 98.28 & 2.04 & 2.699 \\
No version reuse & 40 & 51.18 & 1.06 & 2.852 \\
\bottomrule
\end{tabular}
\end{minipage}\hfill
\begin{minipage}{0.49\textwidth}
\centering
\textbf{(b) \algst{}}\\[2pt]
\begin{tabular}{lrrrr}
\toprule
Variant & $N$ & Time (s) & Rel. & Memory (GB) \\
\midrule
Full & 40 & 62.98 & 1.00 & 4.686 \\
One order group & 34 & 64.93 & 2.55 & 2.269 \\
Central dispatcher & 40 & 124.04 & 1.97 & 11.458 \\
Eager monolithic snapshot & 38 & 145.42 & 9.96 & 5.881 \\
No segment-tree index & 40 & 102.22 & 1.62 & 2.694 \\
No version reuse & 40 & 69.19 & 1.10 & 4.668 \\
\bottomrule
\end{tabular}
\end{minipage}
\end{table*}

\begin{table*}[!t]
\centering
\caption{AS3356, twenty objectives: component results for both index backends.}
\label{tab:ablation-as}
\scriptsize
\begin{minipage}{0.49\textwidth}
\centering
\textbf{(a) \algbm{}}\\[2pt]
\begin{tabular}{lrrrr}
\toprule
Variant & $N$ & Time (s) & Rel. & Memory (GB) \\
\midrule
Full & 40 & 19.17 & 1.00 & 0.700 \\
One order group & 40 & 23.78 & 1.24 & 0.677 \\
Central dispatcher & 40 & 19.86 & 1.04 & 0.708 \\
Eager monolithic snapshot & 40 & 94.98 & 4.95 & 4.306 \\
No block-minimum index & 40 & 30.57 & 1.59 & 0.687 \\
No version reuse & 40 & 20.51 & 1.07 & 0.707 \\
\bottomrule
\end{tabular}
\end{minipage}\hfill
\begin{minipage}{0.49\textwidth}
\centering
\textbf{(b) \algst{}}\\[2pt]
\begin{tabular}{lrrrr}
\toprule
Variant & $N$ & Time (s) & Rel. & Memory (GB) \\
\midrule
Full & 40 & 9.78 & 1.00 & 1.391 \\
One order group & 40 & 10.39 & 1.06 & 1.380 \\
Central dispatcher & 40 & 9.92 & 1.01 & 1.401 \\
Eager monolithic snapshot & 39 & 106.82 & 16.87 & 12.435 \\
No segment-tree index & 40 & 30.02 & 3.07 & 0.694 \\
No version reuse & 40 & 9.88 & 1.01 & 1.385 \\
\bottomrule
\end{tabular}
\end{minipage}
\end{table*}

\begin{table*}[!t]
\centering
\caption{Dense-180, twenty objectives: component results on the five-query ablation subset.}
\label{tab:ablation-dense}
\scriptsize
\begin{minipage}{0.49\textwidth}
\centering
\textbf{(a) \algbm{}}\\[2pt]
\begin{tabular}{lrrrr}
\toprule
Variant & $N$ & Time (s) & Rel. & Memory (GB) \\
\midrule
Full & 5 & 75.15 & 1.00 & 1.258 \\
One order group & 5 & 87.31 & 1.16 & 1.204 \\
Central dispatcher & 5 & 74.83 & 1.00 & 1.286 \\
Eager monolithic snapshot & 5 & 221.91 & 2.95 & 4.886 \\
No block-minimum index & 5 & 91.77 & 1.22 & 1.262 \\
No version reuse & 5 & 75.84 & 1.01 & 1.273 \\
\bottomrule
\end{tabular}
\end{minipage}\hfill
\begin{minipage}{0.49\textwidth}
\centering
\textbf{(b) \algst{}}\\[2pt]
\begin{tabular}{lrrrr}
\toprule
Variant & $N$ & Time (s) & Rel. & Memory (GB) \\
\midrule
Full & 5 & 43.09 & 1.00 & 2.502 \\
One order group & 5 & 46.37 & 1.08 & 2.436 \\
Central dispatcher & 5 & 42.75 & 0.99 & 2.514 \\
Eager monolithic snapshot & 5 & 455.02 & 10.56 & 15.596 \\
No segment-tree index & 5 & 89.61 & 2.08 & 1.267 \\
No version reuse & 5 & 43.30 & 1.00 & 2.511 \\
\bottomrule
\end{tabular}
\end{minipage}
\end{table*}

The most consistent result is the value of base--delta publication. On AS3356 and Dense-180 with twenty objectives, eager monolithic reconstruction is 2.95--4.95$\times$ slower for the block backend and 10.56--16.87$\times$ slower for the segment-tree backend. The penalty is larger for the tree because each accepted label forces a complete interval hierarchy to be rebuilt. Eager reconstruction also raises memory substantially, indicating that repeated full copies and delayed reclamation increase the number and size of simultaneously live snapshots.

The indexes are likewise important. Removing block minima slows the three settings by 1.22--2.04$\times$. Removing the segment tree slows them by 1.62--3.07$\times$. The stronger high-dimensional effect is consistent with longer feasible prefixes and the increased value of range elimination. The no-index variants sometimes use less memory than \algst{}, but the runtime penalty shows that the extra summaries are worthwhile when frontiers are large.

Version reuse has a smaller but measurable effect. The block backend slows by approximately 1--7\%, and the segment-tree backend by approximately 0.5--10\%. Reuse saves a complete second snapshot traversal only when the candidate label reaches its update worker before another accepted label changes that snapshot's version. Longer queue delay or a higher update rate makes a version change more likely, so the update worker performs exact revalidation more often and the benefit decreases. Direct owner delivery is highly important on NYC-Road, where a centralized dispatcher nearly doubles runtime and sharply increases memory, but it is close to neutral on AS3356 and the five-query Dense-180 subset. This topology dependence suggests that the centralized stage becomes harmful when candidate arrival rates create queue contention or delay publication of pruning information. Objective-order diversity has its largest effect on NYC-Road, but it is also material for \algbm{} on AS3356, where one order group is 1.24$\times$ slower, and on Dense-180, where it is 1.16$\times$ slower.

\subsection{Thread Scaling}
\label{subsec:scaling}

We evaluate both index backends on Dense-180 with ten objectives over the same 40 queries while increasing the two worker pools proportionally. The configurations are $8/8$, $16/16$, $24/24$, $32/32$, $40/40$, and $48/48$ search/update workers, for 16--96 workers in total; the update count includes the dedicated goal updater. This balanced series measures end-to-end scaling of the complete two-stage architecture. Because both roles and the vertex-to-update-worker partition change together, it is a system-level scaling study rather than an isolation of either stage.

\begin{table*}[!t]
\centering
\caption{Thread scaling on Dense-180 with ten objectives over 40 queries. Speedup is relative to the corresponding 16-worker result.}
\label{tab:scaling}
\footnotesize
\setlength{\tabcolsep}{3.4pt}
\begin{tabular}{rrrrrrrr}
\toprule
Total workers & Search/Update &
\multicolumn{3}{c}{\algbm{}} &
\multicolumn{3}{c}{\algst{}} \\
& & Time (s) & Speedup & Memory (GB) & Time (s) & Speedup & Memory (GB) \\
\midrule
16 & 8/8   & 7.363 & 1.00 & 0.080 & 7.680 & 1.00 & 0.142 \\
32 & 16/16 & 3.793 & 1.94 & 0.097 & 3.932 & 1.95 & 0.175 \\
48 & 24/24 & 2.594 & 2.84 & 0.112 & 2.689 & 2.86 & 0.202 \\
64 & 32/32 & 1.960 & 3.76 & 0.127 & 2.039 & 3.77 & 0.231 \\
80 & 40/40 & 1.664 & 4.43 & 0.144 & 1.739 & 4.42 & 0.259 \\
96 & 48/48 & 1.501 & 4.90 & 0.165 & 1.587 & 4.84 & 0.295 \\
\bottomrule
\end{tabular}
\end{table*}

\begin{figure*}[!t]
    \centering
    \begin{minipage}{0.49\textwidth}
        \centering
        \includegraphics[width=\linewidth]{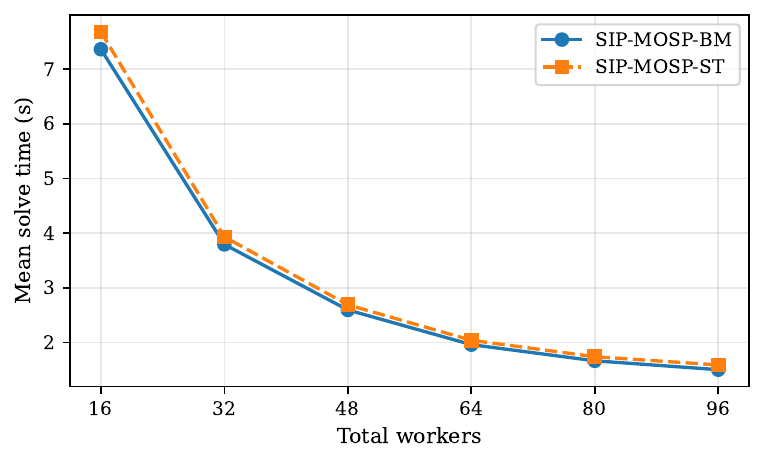}
    \end{minipage}\hfill
    \begin{minipage}{0.49\textwidth}
        \centering
        \includegraphics[width=\linewidth]{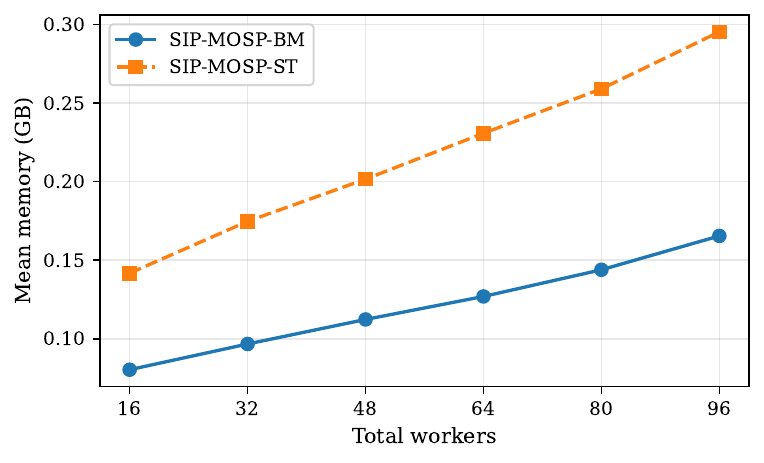}
    \end{minipage}
    \caption{Thread scaling of both SIP-MOSP variants on Dense-180 with ten objectives. Search and update workers are increased in equal numbers.}
    \label{fig:scaling}
\end{figure*}

Both variants improve monotonically as additional workers are introduced. From 16 to 96 workers, \algbm{} decreases from 7.363 to 1.501~s, a 4.90$\times$ speedup, and \algst{} decreases from 7.680 to 1.587~s, a 4.84$\times$ speedup. Mean memory grows from 0.080 to 0.165~GB for \algbm{} and from 0.142 to 0.295~GB for \algst{}, approximately 2.06$\times$ and 2.08$\times$, respectively. Thus, a sixfold increase in workers yields nearly a fivefold reduction in solve time while memory roughly doubles rather than increasing in proportion to the worker count. Gains become smaller beyond 64 workers, reflecting scheduling, communication, and frontier-update overheads as well as the changing partition among update workers. The experiment is limited to this representative Dense-180 workload and is intended to demonstrate that both backends can exploit substantially more workers than the ten objective orders alone would provide.

\subsection{Illustrative CPU-Utilization Trace}

\noindent\begin{minipage}{\columnwidth}
Fig.~\ref{fig:cpu} reports equivalent busy cores for one Dense-180 twenty-objective query, where one unit denotes one fully busy logical core. Table~\ref{tab:cpu} summarizes the same process-group traces.
\end{minipage}

\begin{figure*}[!t]
    \centering
    \includegraphics[width=0.83\textwidth]{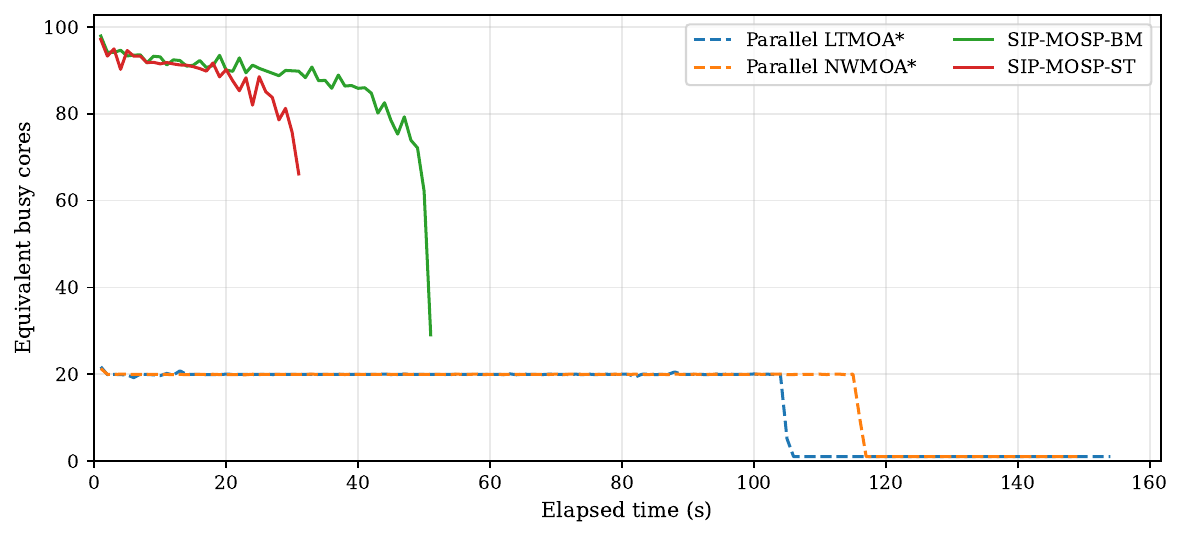}
    \caption{Equivalent busy CPU cores for one Dense-180 twenty-objective query, measured over each algorithm's complete process group.}
    \label{fig:cpu}
\end{figure*}

\begin{table}[!t]
\centering
\caption{Summary of the illustrative CPU-utilization trace.}
\label{tab:cpu}
\footnotesize
\begin{tabular}{lrr}
\toprule
Algorithm & Runtime (s) & Mean busy cores \\
\midrule
Parallel LTMOA* & 154.6 & 13.81 \\
Parallel NWMOA* & 149.4 & 15.69 \\
\algbm{} & 51.8 & 86.09 \\
\algst{} & 31.9 & 86.60 \\
\bottomrule
\end{tabular}
\end{table}

The traces contrast the two parallelization models on a dense, twenty-objective workload. Parallel LTMOA* and parallel NWMOA* use one search per objective order. They begin near their 20 available search threads, but utilization falls sharply as individual searches finish and the remaining runtime is determined by a small number of unfinished searches. \algname{} instead distributes label expansion and vertex-frontier maintenance within one cooperative search. Both variants sustain approximately 86 busy cores on average and avoid the long low-utilization tail, while completing substantially earlier. This example shows that \algname{} can sustain high parallel activity throughout a difficult high-dimensional run, whereas the objective-order baselines lose active searches as their independent runs finish. The multi-query tables provide the aggregate runtime and memory evidence.

\subsection{Cross-Paper Comparison with OPMOS and MPMOS}

Table~\ref{tab:gpu} compares supplied results on the same five NYC-Road four-objective queries used by OPMOS and MPMOS. The sequential column is NAMOA*. SIP-MOSP-BM and OPMOS both have a 72-thread result, and we additionally report SIP-MOSP-BM with 96 workers. All methods use the same topology and query identifiers. The implementations and machines were evaluated separately, so the table is not an identical-hardware experiment; nevertheless, the shared benchmark and the matched 72-thread CPU result permit an informative comparison of the efficiency achieved by the published designs. Parentheses give the reported speedup over NAMOA*. The geometric mean is retained because the original comparison summarizes multiplicative speedups across queries.

\begin{table*}[!t]
\centering
\caption{Cross-paper NYC-Road four-objective times in seconds. Parentheses give speedup over NAMOA*.}
\label{tab:gpu}
\scriptsize
\setlength{\tabcolsep}{3.5pt}
\begin{tabular}{lrrrrr}
\toprule
Query & NAMOA* & OPMOS~\cite{Gold2025} & MPMOS~\cite{Gold2026MPMOS} & \algbm{} (72) & \algbm{} (96) \\
\midrule
NY6  & 3,788.1 & 1,124.0 (3.4$\times$) & 5.190 (720$\times$) & 4.621 (820$\times$) & \textbf{2.875} (1,318$\times$) \\
NY15 & 2.94 & 0.800 (3.7$\times$) & \textbf{0.134} (21.9$\times$) & 0.174 (16.9$\times$) & 0.161 (18.3$\times$) \\
NY20 & 399.1 & 128.1 (3.1$\times$) & 1.155 (346$\times$) & 1.028 (388$\times$) & \textbf{0.873} (457$\times$) \\
NY31 & 3,041.0 & 226.4 (13$\times$) & \textbf{3.346} (909$\times$) & 9.108 (334$\times$) & 8.258 (368$\times$) \\
NY50 & 39.26 & 8.551 (4.6$\times$) & 0.316 (124$\times$) & 0.333 (118$\times$) & \textbf{0.309} (127$\times$) \\
\midrule
Geometric mean & 221.2 & 46.7 (4.7$\times$) & \textbf{0.968} (229$\times$) & 1.202 (184$\times$) & 1.006 (220$\times$) \\
\bottomrule
\end{tabular}
\end{table*}

The matched 72-thread comparison reduces the geometric-mean time from 46.7~s for OPMOS to 1.202~s for \algbm{}, a 38.9$\times$ reduction. Increasing SIP-MOSP-BM to 96 workers lowers the geometric mean further to 1.006~s, 46.4$\times$ below OPMOS and only about 4\% above the 0.968~s reported for GPU-based MPMOS. The 96-worker CPU result is faster than MPMOS on three of the five queries. These cross-paper data therefore reinforce the efficiency of the proposed organization: the shared-memory CPU implementation substantially improves on the reported OPMOS times and reaches the same practical performance range as the massively parallel GPU design on this subset. Because the systems were evaluated separately, the table complements rather than replaces the controlled comparisons in Tables~\ref{tab:nyc-main}--\ref{tab:dense-main}.

\subsection{Discussion}

Parallel exact MOSP remains a recognized research challenge. Salzman \emph{et al.}~\cite{Salzman2023Survey} list parallelization among the open challenges and research opportunities for heuristic MOSP. The later survey of multi-objective search~\cite{Salzman2026Emerging} likewise notes that parallel MOS research has remained largely unexplored and that most existing algorithms scale poorly beyond two or three objectives. Only a small number of exact methods have reported practical parallel evaluation, and recent empirical studies have often centered on road-network benchmarks. In this context, our results show that one cooperative shared-memory architecture can operate effectively on a road graph, an ISP topology, and a complete directed graph, including high-dimensional settings in which both runtime and memory become limiting factors.

\paragraph{Sources of scalability}
The runtime and CPU traces show that the main benefit is not simply adding threads. Assigning each vertex to one update worker prevents concurrent writes to the same frontier, while different workers can update unrelated vertex frontiers simultaneously. Immutable snapshots let search workers execute dominance pruning without acquiring read locks on mutable frontiers. The incoming-queue/private-priority-queue pair separates concurrent delivery from ordered local processing, and direct owner delivery avoids an additional dequeue, ownership lookup, and enqueue at a centralized dispatcher. Together, these mechanisms shorten the interval between discovery of a useful vector and publication of the pruning information it creates.

\paragraph{Why base--delta snapshots matter}
Exact search creates a feedback loop: accepting a vector changes a frontier, and the changed frontier should quickly prune other labels. Rebuilding a complete indexed snapshot after every accepted vector preserves freshness but makes publication proportional to the full frontier size. Delaying publication reduces construction cost but weakens pruning. The base--delta organization resolves this tension by publishing each recent vector in a small delta while retaining an immutable indexed base. The ablations show that this is not a minor implementation detail. Eager monolithic reconstruction causes the largest and most consistent slowdowns, especially for the segment-tree backend, whose complete index contains substantially more summary state.

\paragraph{Choosing between block and segment-tree minima}
The two backends occupy different points in the time--memory trade-off. Block minima use fewer summaries and sequential memory, which is advantageous on NYC-Road and the ten-objective Dense-180 setting. Segment-tree minima offer finer interval elimination and become fastest on AS3356 and Dense-180 with twenty objectives. The no-index ablations show that the advantage is not merely low-level tuning: removing either index increases runtime, and the segment-tree penalty reaches 3.07$\times$ on AS3356. A practical extension is an adaptive backend that selects blocks or a tree from measured frontier size, feasible-prefix length, index hit rate, and memory pressure.

\paragraph{Shared search state and memory behavior}
Parallel LTMOA* and NWMOA* execute multiple complete objective-order searches, which duplicates open lists and node-frontier state. \algname{} uses objective orders only as local scheduling policies inside one cooperative search. The high-dimensional memory reductions, 10.9$\times$ on AS3356 and 60.3$\times$ on Dense-180 relative to parallel NWMOA*, follow from this structural difference. NYC-Road behaves differently. Its low-dimensional frontiers are comparatively compact, while start--goal routes can traverse many intermediate vertices. The asynchronous pipeline may therefore keep many accepted labels, incoming entries, local queue entries, and path-reconstruction records live at once. When frontier replication is inexpensive, this concurrency and snapshot metadata can outweigh the memory saved by sharing one search state. The method's strongest memory advantage consequently appears when high-dimensional Pareto state, rather than path depth alone, dominates storage.

\paragraph{Topology sensitivity and scheduling}
The dispatcher and order-group ablations are topology dependent. Direct owner delivery is decisive on NYC-Road but nearly neutral on AS3356 and the five-query Dense-180 subset. A centralized dispatcher becomes costly when many search workers simultaneously generate candidates and the extra queueing delay postpones useful frontier updates. Objective-order diversity is beneficial when one lexicographic order expands an unfavorable region of the search space; the 1.24$\times$ AS3356 slowdown for one order group shows that this effect is not confined to road networks. Version reuse depends on a different timing relation. It benefits labels that reach their update worker before the relevant snapshot changes, so its gain reflects queue delay, update frequency, and dominance-query cost rather than topology alone.

\paragraph{Relation to other parallel models}
Objective-order portfolios and SIP-MOSP expose complementary forms of parallelism. Portfolios obtain independent best-first searches whose progress can diversify early solution discovery, but their primary concurrency is the number of active orders and their state is replicated. SIP-MOSP shares canonical frontiers and distributes labels and vertex-local updates within one search. OPMOS and MPMOS instead rely on ordered accelerator-oriented execution. The results suggest that shared-memory CPUs can exploit substantial parallelism without either replicating complete searches or abandoning exact frontier admission.

\paragraph{Limits and future directions}
The current implementation assumes nonnegative additive edge costs so that reverse Dijkstra lower bounds are available. The snapshot and ownership architecture is independent of this choice, and a bounded negative-weight extension could use the lower-bound machinery of NWMOA*. The present vertex assignment is static; NUMA-aware graph partitioning, cost-aware ownership, or migration of heavily loaded vertex groups may improve locality and balance. The snapshot indexes are also static within a consolidated base. Future work includes adaptive block size, automatic backend selection, vectorized scans inside surviving blocks, topology-aware order groups, and GPU or heterogeneous snapshot queries.

\section{Conclusion}
\label{sec:conclusion}

This work develops an exact shared-memory MOSP framework around a separation of responsibilities. Search workers preserve local best-first processing and generate candidates, while update workers provide exclusive admission and Pareto maintenance for assigned vertex frontiers. Immutable versioned snapshots connect the two stages: they expose current pruning information without allowing search threads to access mutable frontier state.

The main algorithmic contribution combines base--delta publication, exact indexed pruning, direct owner delivery, and version-aware validation. Block-minimum and segment-tree-minimum summaries remove ranges that cannot contain a dominator, while exact comparisons preserve semantics. The delta keeps newly accepted vectors visible without rebuilding a large index after every update, and the version field determines when a search-side negative query can be reused safely. These mechanisms retain exact cost-unique output under asynchronous scheduling.

The experiments show that the organization is effective across road, ISP, and dense graph settings. It improves on the state-of-the-art sequential and parallel exact baselines in all reported main settings, provides particularly large time and memory gains at high objective counts, and maintains high processor utilization. The ablations identify base--delta publication and indexed pruning as the most consistent contributors, while direct delivery, order diversity, and version reuse respond to workload structure. Overall, SIP-MOSP provides a practical route to scaling one exact cooperative search beyond the number of objectives without replicating the complete Pareto search state.

\bibliographystyle{IEEEtran}
\bibliography{references}

\end{document}